\documentclass[11pt]{article}

\usepackage[left=2.2cm,right=2.2cm,top=2.5cm,bottom=2.5cm]{geometry}

\usepackage{authblk}

\title{A coalgebraic higher-order modal fixed-point logic}

\author{Ryan Tay, Harsh Beohar, and Charles Grellois}
\affil{University of Sheffield, UK}

\date{}

\usepackage[english]{babel}
\usepackage[T1]{fontenc}
\usepackage[utf8]{inputenc}
\usepackage{graphicx}
\usepackage{framed}
\usepackage{amsmath}
\usepackage{amssymb}
\usepackage{amsfonts}
\usepackage{mathrsfs}
\usepackage{array}
\usepackage{amsthm}
\usepackage{amscd}
\usepackage{mathtools}
\usepackage{enumitem}
\usepackage{bbm}
\usepackage[normalem]{ulem}
\usepackage{cancel}
\usepackage[round]{natbib}
\usepackage{datetime}
\usepackage{pifont}

\usepackage{stmaryrd}
\usepackage{cmll}

\usepackage{url}
\usepackage{hyperref}
\usepackage[capitalise, nameinlink]{cleveref}
\usepackage{breakcites}

\usepackage{xcolor}
\hypersetup{
    colorlinks=true,
    breaklinks=true,
    linkcolor={red!70!black},
    citecolor={green!45!black},
    urlcolor={blue!80!black}
}

\newtheoremstyle{break_italics}
  {\topsep}{\topsep}%
  {\itshape}{}%
  {\bfseries}{}%
  {\newline}{}%
\theoremstyle{break_italics}
\newtheorem{theorem}{Theorem}[section]
\newtheorem*{theorem*}{Theorem}

\newtheorem*{lemma*}{Lemma}
\newtheorem{proposition}[theorem]{Proposition}
\newtheorem*{proposition*}{Proposition}

\newtheorem*{corollary*}{Corollary}
\newtheorem{definition}[theorem]{Definition}
\newtheorem*{definition*}{Definition}

\newtheorem*{claim*}{Claim}

\newtheorem*{fact*}{Fact}
\newtheorem{question}[theorem]{Question}
\newtheorem*{question*}{Question}

\newtheoremstyle{break_upright}
  {\topsep}{\topsep}%
  {}{}%
  {\bfseries}{}%
  {\newline}{}%
\theoremstyle{break_upright}
\newtheorem{example}[theorem]{Example}
\newtheorem*{observation*}{Observation}

\theoremstyle{remark}
\newtheorem*{remark}{Remark}

\usepackage{tikz}
\usepackage{tikz-cd}
\usepackage{adjustbox}
\usetikzlibrary{decorations.pathmorphing, arrows.meta}
\usepackage{quiver}

\usepackage{bussproofs}
\newenvironment{bprooftree}
  {\leavevmode\hbox\bgroup}
  {\DisplayProof\egroup}
 
\allowdisplaybreaks

\newcommand{\dom}{\mathrm{dom}}
\newcommand{\ord}{\mathrm{ord}}
\newcommand{\LFP}{\mathrm{LFP}}
\newcommand{\GFP}{\mathrm{GFP}}
\newcommand{\Set}{\mathbf{Set}}
\newcommand{\op}{\mathrm{op}}
\newcommand{\HFL}{\mathsf{HFL}}
\newcommand{\Mod}{\mathrm{Mod}}
\newcommand{\B}{\mathbb{B}}
\newcommand{\prob}{\mathbb{P}}

\newcommand{\suppt}{\mathrm{suppt}}
\newcommand{\Coalg}{\mathbf{Coalg}}
\newcommand{\id}{\mathrm{id}}

\newcommand{\cmark}{\text{\ding{51}}}

\renewcommand{\P}{\mathcal{P}}
\renewcommand{\S}{\mathbb{S}}

\newcommand{\Circle}{\bigcirc}

\begin{document}

\maketitle

\begin{abstract}
	We introduce a coalgebraic extension of the higher-order modal fixed-point logic (HFL) which subsumes both HFL and its probabilistic extension. We show that the emptiness problem for nondeterministic finite automata as well as the value-1 problem for probabilistic automata reduce to model-checking problems for this coalgebraic formulation of HFL.
\end{abstract}

\tableofcontents

\section{Introduction}

Studying the modal logic of a coalgebra of an endofunctor on the category $\Set$, through the use of predicate liftings as introduced by \citet{Pattinson2003}, has proven to be a very fruitful line of research. It has allowed us to study the modal logic of many different notions of transition systems (see, for instance, \citet{JacobsIntroductionToCoalgebra} or \citet{RuttenTheMethodOfCoalgebra}) in a uniform manner, encapsulating modalities beyond the classical $\Box$ and $\Diamond$. Many substantial results have been obtained concerning coalgebraic modal logic, for example, its finite model property \cite[Theorem~6.3]{Pattinson2003}, its expressivity \cite[Theorem~14]{Schoreder2008}, and its uniform interpolation \cite*[Theorem~37]{SeifanSchroderPattinson2017}. Furthermore, these semantics have been extended by \citet*{CirsteaKupkePattinson2011} to encapsulate the modal $\mu$-calculus. This extension also enjoys the finite (in fact, bounded) model property \cite*[Theorem~3]{FontaineLealVenema2010}. Moreover, a celebrated result by \cite{JaninWalukiewicz1996}, asserting that the modal $\mu$-calculus is the bisimulation-invariant fragment of monadic second-order logic, has been lifted to the coalgebraic setting under certain conditions \cite*[Theorem~5.10]{EnqvistSeifanVenema2017}.

\cite{ViswanathanViswanathan2004} introduced another extension of the modal $\mu$-calculus which allows for fixed-points of higher-order functions. They called this logic the \emph{higher-order modal fixed-point logic (HFL)}. Aside from being an interesting logic to study in its own right, due to it subsuming the modal $\mu$-calculus, \citet*[Theorem~10 and Theorem~22]{KobayashiLozesBruse2017} established mutual reductions between the model-checking problem for HFL and the higher-order model-checking problem (see, for instance, \cite{Ong2015}). As the decidability of the higher-order model-checking problem is not so easy to see \citep{Ong2006} \citep{KobayashiOng2009} \citep{SalvatiWalukiewicz2014}, whereas the decidability of the HFL model-checking problem can be spotted immediately from definitions, serious effort has been devoted to the HFL model-checking problem (see, for instance, \cite{Kobayashi2021}).

The logic HFL has been given a further probabilistic extension by \cite*{MitaniKobayashiTsukada2021}. This \emph{probabilistic higher-order modal fixed-point logic (PHFL)} also has connections to probabilistic higher-order recursion schemes \citep*[Definition~2.1]{KobayashiDalLagoGrellois2020} as well as recursive Markov chains \citep[Section~2.1]{EtessamiYannakakis2009}; such connections were established in the same paper by \citet*[Section~5.3]{MitaniKobayashiTsukada2021}. While the original formulation of HFL interpreted its formulas in labelled transition systems, PHFL formulas find their semantics in Markov chains.

We take the natural step of providing coalgebraic semantics to the higher-order modal fixed-point logic, with the intention of capturing both the logics HFL and PHFL. Following the coalgebraic $\mu$-calculus \citep*{CirsteaKupkePattinson2011}, we restrict our attention to coalgebras of endofunctors on $\Set$. As an application, we show that this coalgebraic HFL is able to express a certain generalisation of both the emptiness problem of nondeterministic finite automata (which asks whether every word is rejected, or equivalently, whether their associated regular language is empty) as well as the value-1 problem (see, for instance, \citet[Problem~5]{GimbertOualhadj2010}) of probabilistic automata in the sense of \citet[Definition~4]{Rabin1963}.

\section{Syntax and semantics}

The coalgebraic higher-order modal fixed-point logic (henceforth \emph{coalgebraic HFL}) will be an extension of the classical modal logic in the following way: each modality will be given by a \emph{monotone predicate lifting} (which will be defined later); operations from the simply-typed $\lambda$-calculus will be introduced; and least and greatest fixed-point operators will be available for higher-order functions.

We fix a \emph{base type} $\bullet$. The set $\mathrm{Types}$ of \emph{(simple) types} is generated by the grammar
\[
	\tau \Coloneqq \bullet \ | \ \tau \to \tau.
\]
For types $\tau_1, \tau_2, \dots, \tau_n, \tau_{n+1}$, we write
\[
	\tau_1 \to \tau_2 \to \cdots \to \tau_n \to \tau_{n+1} \coloneqq \tau_1 \to (\tau_2 \to (\cdots \to (\tau_n \to \tau_{n+1}) \cdots)),
\]
so that $\to$ associates to the right. The \emph{order} of a type $\tau$, denoted $\ord(\tau)$, is defined recursively as follows:
\begin{align*}
	\ord(\bullet) &\coloneqq 0, \\
	\ord(\tau_1 \to \tau_2) &\coloneqq \max(\mathrm{ord}(\tau_1) + 1, \mathrm{ord}(\tau_2)) \quad \text{for all types $\tau_1$ and $\tau_2$.}
\end{align*}
Note that
\[
	\ord(\tau_1 \to \cdots \to \tau_n \to \bullet) = 1 + \max(\ord(\tau_1), \dots, \ord(\tau_n))
\]
for types $\tau_1, \dots, \tau_n$.

Fix a complete lattice $\B$, i.e. a poset with all small meets and small joins. Let $\top_\B$ and $\bot_\B$ respectively denote the maximum and minimum elements of $\B$, and let $<_\B$ and $\leq_\B$ respectively denote the strict and weak partial orders on $\B$. Furthermore, we fix a countably infinite set $\mathrm{Var}$ of \emph{variables} and a set $\Mod$ of \emph{modalities}, where each modality $\heartsuit \in \Mod$ has an \emph{arity} $n \in \omega$. Without loss of generality, we shall assume that the sets $\B$, $\mathrm{Var}$, and $\Mod$ are all pairwise disjoint.

\begin{definition} \label{defn_syntax}
	The set of \emph{formulas} in $\HFL_{\Mod}^\B$ is generated by the grammar
	\[
		\varphi \Coloneqq \top \ | \ \bot \ | \ X \ | \ \varphi \land \varphi \ | \ \varphi \lor \varphi \ | \ \heartsuit(\underbrace{\varphi, \dots, \varphi}_{\text{$n$ copies of $\varphi$}}) \ | \ \lambda X^\tau.\varphi \ | \ \varphi\varphi \ | \ \mu X^\tau.\varphi \ | \ \nu X^\tau.\varphi,
	\]
	where $X$ ranges over $\mathrm{Var}$, $\heartsuit$ ranges over $\Mod$ and has arity $n \in \omega$, and $\tau$ ranges over $\mathrm{Types}$. We identify formulas up to $\alpha$-equivalence. For formulas $\varphi_1, \varphi_2, \varphi_3, \dots, \varphi_n \in \HFL_{\Mod}^\B$, we write
	\[
		\varphi_1\varphi_2\varphi_3\dots\varphi_n \coloneqq (\dots((\varphi_1\varphi_2)\varphi_3)\dots)\varphi_n
	\]
	so that $\lambda$-application associates to the left.
\end{definition}

Note that \cref{defn_syntax} above makes no use of the complete lattice $\B$. Instead, this complete lattice will make an appearance later when providing semantics to formulas in \cref{defn_semantics}. When providing these semantics, we will also be in a setting where the modalities in $\Mod$ depend on $\B$.

The semantics of a formula of the form $\lambda X^\tau.\varphi$ intuitively corresponds to a function which takes an input $X$ of type $\tau$ and returns $\varphi$ (where $X$ may be a free variable in $\varphi$). A formula of the form $\varphi\psi$ intuitively corresponds to the $\lambda$-application of $\varphi$ on input $\psi$. If the semantics of a formula $\lambda X^\tau.\varphi$ corresponds to a function which takes an input $X$ of type $\tau$ and returns an output $\varphi$ of the same type $\tau$, then the semantics of $\mu X^\tau.\varphi$ (resp. $\nu X^\tau.\varphi$) corresponds to the least (resp. greatest) fixed-point of the semantics of $\lambda X^\tau.\varphi$. The semantics provided later will ensure that these are actually least and greatest fixed-points.

The original formulation of HFL \citep{ViswanathanViswanathan2004} included negations. We have chosen to exclude negations, following a negation elimination result by \citet[Corollary~5]{Lozes2015} (see also \citet[Sections~3.1.1 and 3.1.2]{BrusePhDthesis2018}). This approach was adopted in, for instance, \citet*[Section~2.3]{KobayashiLozesBruse2017}.

\begin{definition} \label{defn_typing}
	A \emph{context} is a partial function $\Gamma \colon \mathrm{Var} \rightharpoonup \mathrm{Types}$ whose domain is finite. The \emph{type judgement relation} $\vdash$ is the smallest relation closed under the following derivation rules: for all contexts $\Gamma \colon \mathrm{Var} \rightharpoonup \mathrm{Types}$, all variables $X \notin \dom(\Gamma)$, all modalities $\heartsuit \in \Mod$ of arity $n \in \omega$, all formulas $\varphi,\psi, \varphi_1, \dots, \varphi_n \in \HFL_{\Mod}^\B$, and all types $\tau, \tau_1, \tau_2 \in \mathrm{Types}$,
	\[
		\begin{bprooftree}
			\AxiomC{}
			\RightLabel{\scriptsize $\top$}
			\UnaryInfC{$\Gamma \vdash \top : \bullet$}
		\end{bprooftree}
		\qquad \qquad
		\begin{bprooftree}
			\AxiomC{}
			\RightLabel{\scriptsize $\bot$}
			\UnaryInfC{$\Gamma \vdash \bot : \bullet$}
		\end{bprooftree}
		\qquad \qquad
		\begin{bprooftree}
			\AxiomC{}
			\RightLabel{\scriptsize $\mathrm{var}$}
			\UnaryInfC{$\Gamma, X:\tau \vdash X:\tau$}
		\end{bprooftree} 
	\] \\
	\[
		\begin{bprooftree}
			\AxiomC{$\Gamma \vdash \varphi : \bullet$}
			\AxiomC{$\Gamma \vdash \psi : \bullet$}
			\RightLabel{\scriptsize $\land$}
			\BinaryInfC{$\Gamma \vdash \varphi \land \psi : \bullet$}
		\end{bprooftree}
		\qquad \qquad \qquad
		\begin{bprooftree}
			\AxiomC{$\Gamma \vdash \varphi : \bullet$}
			\AxiomC{$\Gamma \vdash \psi : \bullet$}
			\RightLabel{\scriptsize $\lor$}
			\BinaryInfC{$\Gamma \vdash \varphi \lor \psi : \bullet$}
		\end{bprooftree}
	\] \\
	\[
		\begin{bprooftree}
			\AxiomC{$\Gamma \vdash \varphi_1 : \bullet$}
			\AxiomC{$\cdots$}
			\AxiomC{$\Gamma \vdash \varphi_n : \bullet$}
			\RightLabel{\scriptsize $\heartsuit$}
			\TrinaryInfC{$\Gamma \vdash \heartsuit(\varphi_1, \dots, \varphi_n) : \bullet$}
		\end{bprooftree}
	\] \\
	\[
		\begin{bprooftree}
			\AxiomC{$\Gamma,X:\tau_1 \vdash \varphi : \tau_2$}
			\RightLabel{\scriptsize $\lambda\mathrm{abs}$}
			\UnaryInfC{$\Gamma \vdash \lambda X^{\tau_1}.\varphi : \tau_1 \to \tau_2$}
		\end{bprooftree}
		\qquad \qquad \qquad 
		\begin{bprooftree}
			\AxiomC{$\Gamma \vdash \varphi : \tau_1 \to \tau_2$}
			\AxiomC{$\Gamma \vdash \psi : \tau_1$}
			\RightLabel{\scriptsize $\lambda\mathrm{app}$}
			\BinaryInfC{$\Gamma \vdash \varphi\psi : \tau_2$}
		\end{bprooftree}
	\] \\
	\[
		\begin{bprooftree}
			\AxiomC{$\Gamma, X:\tau \vdash \varphi : \tau$}
			\RightLabel{\scriptsize $\mu$}
			\UnaryInfC{$\Gamma \vdash \mu X^\tau.\varphi : \tau$}
		\end{bprooftree}
		\qquad \qquad \qquad
		\begin{bprooftree}
			\AxiomC{$\Gamma, X:\tau \vdash \varphi : \tau$}
			\RightLabel{\scriptsize $\nu$}
			\UnaryInfC{$\Gamma \vdash \nu X^\tau.\varphi : \tau$}
		\end{bprooftree}
	\]
	where $\Gamma, X:\tau$ is the partial function $\Gamma \cup \{X \mapsto \tau\}$.
	
	We say that a formula $\varphi$ is \emph{well-typed} if there exists a context $\Gamma$ and a type $\tau$ such that $\Gamma \vdash \varphi : \tau$ is derivable. A formula $\varphi$ is said to be \emph{closed} if $\varphi$ has no free variables. Furthermore, we say that a closed formula $\varphi$ \emph{has type $\tau$} if $\varnothing \vdash \varphi : \tau$ is derivable.
\end{definition}

An analogous ``Unique Types'' result as in \citet[Proposition~2]{ViswanathanViswanathan2004} can be obtained for the type judgement system of $\HFL_\Mod^\B$: an induction on the height of derivation trees shows that if $\Gamma$ is a context, $\varphi$ is a formula, $\tau$ and $\tau'$ are types, and both $\Gamma \vdash \varphi : \tau$ and $\Gamma \vdash \varphi : \tau'$ are derivable, then $\tau = \tau'$. Furthermore, if $\Gamma \vdash \varphi : \tau$ is derivable, then there is a unique derivation of $\Gamma \vdash \varphi : \tau$ up to $\alpha$-equivalence.

\begin{definition}
	For a closed and well-typed formula $\varphi$, we define its \emph{order}, denoted $\ord(\varphi)$ to be
	\[
		\ord(\varphi) \coloneqq \max\{\,\ord(\tau) \mid \text{$\tau$ is the type of a closed subformula of $\varphi$}\,\}.
	\]
\end{definition}

In particular, order-0 formulas in $\HFL_\Mod^\B$ can be identified with the closed formulas in the negation-free fragment of the coalgebraic $\mu$-calculus introduced by \cite*{CirsteaKupkePattinson2011}.

\begin{example}
	The following are examples of closed and well-typed formulas in $\HFL_\Mod^\B$ together with their orders, where $\heartsuit \in \Mod$ has arity 2.
	\begin{align*}
		\ord(\mu X^\bullet.(X \land \heartsuit(\top, X))) &= 0, \\
		\ord(\mu X^{\bullet \to \bullet}.(\lambda Y^\bullet.(XY))) &= 1, \\
		\ord(\lambda X^{\bullet \to \bullet}.\bot) &= 2, \\
		\ord((\lambda X^{\bullet \to \bullet}.\bot)(\lambda Y^\bullet.Y)) &= 2. \tag*{\qed}
	\end{align*}
\end{example}

We will now provide coalgebraic semantics for formulas in $\HFL_\Mod^\B$. Let $F \colon \Set \to \Set$ be a functor. An \emph{$F$-coalgebra} is a pair $\S = (S, \sigma)$ consisting of a set $S$ and a function $\sigma \colon S \to FS$. We call the set $S$ its \emph{state space}, and we call the function $\sigma \colon  S \to FS$ its \emph{structure map}. For two $F$-coalgebras $\S = (S, \sigma)$ and $\S' = (S', \sigma')$, an \emph{$F$-coalgebra morphism} from $\S$ to $\S'$ is a function $f \colon S \to S'$ such that $(Ff)\sigma = \sigma'f$. Thus, we have a category $\Coalg(F)$ with $F$-coalgebras as objects and $F$-coalgebra morphisms as morphisms. A \emph{pointed $F$-coalgebra} is an $F$-coalgebra $\S = (S, \sigma)$ together with a specified element $s_0 \in S$.

For a natural number $n \in \omega$, an \emph{$n$-ary predicate lifting for $F$ with respect to $\B$} is a natural transformation $\heartsuit \colon (\hom_\Set(-, \B))^n \Rightarrow \hom_\Set (F(-), \B)$, where $\hom_\Set(-, \B) \colon \Set^\op \to \Set$ is the contravariant representable functor with (the underlying set of) $\B$ as its representing object. Note that we allow for $n = 0$; \emph{nullary predicate liftings} are natural transformations from the constant functor $\{*\}$ to the functor $\hom_\Set(F(-), \B)$. An $n$-ary predicate lifting $\heartsuit \colon (\hom_\Set(-, \B))^n \Rightarrow \hom_\Set(F(-), \B)$ is said to be \emph{monotone} if: for all sets $S$ and all $f_1, \dots, f_n, g_1, \dots, g_n \in \hom_\Set(S, \B)$ satisfying $f_i(s) \leq_\B g_i(s)$ for all $s \in S$ and $i \in \{1, \dots, n\}$, we have $(\heartsuit_S(f_1, \dots, f_n))(t) \leq_\B (\heartsuit_S(g_1, \dots, g_n))(t)$ for all $t \in FS$.

For the remainder of this section, we fix a functor $F \colon \Set \to \Set$, a complete lattice $\B$, and a set $\Mod$ of monotone predicate liftings for $F$ with respect to $\B$. When we speak of \emph{modalities}, we mean monotone predicate liftings in $\Mod$.

\begin{definition}
	Fix an $F$-coalgebra $\S = (S, \sigma)$. Let $S_\bullet^\B \coloneqq \hom_\Set(S, \B)$ be the set of all functions from $S$ to $\B$. Equip $S_\bullet^\B$ with its usual partial order: for $f, g \in S_\bullet^\B$, declare $f \leq_{S_\bullet^\B} g$ if and only if $f(s) \leq_\B g(s)$ for all $s \in S$. For types $\tau_1$ and $\tau_2$, define
	\[
		S_{\tau_1 \to \tau_2}^\B \coloneqq \big\{\,f \in \hom_\Set(S_{\tau_1}^\B, S_{\tau_2}^\B) \mid \text{$f$ is (weakly) monotonically increasing}\,\big\},
	\]
	and equip $S_{\tau_1 \to \tau_2}^\B$ with the following partial order: for $f,g \in S_{\tau_1 \to \tau_2}^\B$, stipulate that $f \leq_{S_{\tau_1 \to \tau_2}^\B} g$ whenever $f(u) \leq_{S_{\tau_2}^\B} g(u)$ for all $u \in S_{\tau_1}^\B$. For a type $\tau$ and a function $f \in S_{\tau \to \tau}^\B$, let $\LFP_{S_\tau^\B}(f)$ and $\GFP_{S_\tau^\B}(f)$ respectively denote the least and greatest fixed points of $f$, that is,
	\[
		\LFP_{S_\tau^\B}(f) \coloneqq \inf\{\,u \in S_\tau^\B \mid f(u) \leq_{S_\tau^\B} u\,\} \quad \text{and} \quad \GFP_{S_\tau^\B}(f) \coloneqq \sup\{\,u \in S_\tau^\B \mid u \leq_{S_\tau^\B} f(u)\,\},
	\]
	where $\inf(-)$ and $\sup(-)$ respectively denote the infimum and supremum operations in $S_\tau^\B$.
\end{definition}

\begin{definition} \label{defn_semantics}
	Fix an $F$-coalgebra $\S = (S, \sigma)$.	 For a context $\Gamma = \{X_0 : \tau_0, \dots, X_k : \tau_k\}$, define $\llbracket \Gamma \rrbracket_\S^\B$ to be the set of all functions $\eta \colon \{X_0, \dots, X_k\} \to \bigcup_{i \in \{0,\dots,k\}} S_{\tau_i}^\B$ such that $\eta(X_i) \in S_{\tau_i}^\B$ for all $i \in \{0,\dots,k\}$. For a type judgement $\Gamma \vdash \varphi \colon \tau$, define the function $\llbracket \Gamma \vdash \varphi \colon \tau \rrbracket_\S^\B \colon \llbracket \Gamma \rrbracket_\S^\B \to S_\tau^\B$ by induction on $\varphi$ as follows.
	\begin{align*}
		\llbracket \Gamma \vdash \top : \bullet \rrbracket_\S^\B(\eta) &\coloneqq \top_{S_\bullet^\B}, \\
		\llbracket \Gamma \vdash \bot : \bullet \rrbracket_\S^\B(\eta) &\coloneqq \bot_{S_\bullet^\B}, \\
		\llbracket \Gamma, X : \tau \vdash X : \tau \rrbracket_\S^\B(\eta) &\coloneqq \eta(X), \\
		\llbracket \Gamma \vdash \varphi \land \psi : \bullet \rrbracket_\S^\B(\eta) &\coloneqq \inf\big\{\llbracket \Gamma \vdash \varphi : \bullet \rrbracket_\S^\B(\eta), \ \llbracket \Gamma \vdash \psi : \bullet \rrbracket_\S^\B(\eta)\big\}, \\
		\llbracket \Gamma \vdash \varphi \lor \psi : \bullet \rrbracket_\S^\B(\eta) &\coloneqq \sup\big\{\llbracket \Gamma \vdash \varphi : \bullet \rrbracket_\S^\B(\eta), \, \llbracket \Gamma \vdash \psi : \bullet \rrbracket_\S^\B(\eta)\big\}, \\
		\llbracket \Gamma \vdash \heartsuit(\varphi_1, \dots, \varphi_n) : \bullet \rrbracket_\S^\B(\eta) 
				&\coloneqq \hom_\Set(\sigma, \B)\Big(\heartsuit_S\big(\llbracket \Gamma \vdash \varphi_1 \colon \bullet\rrbracket_\S^\B(\eta), \dots, \llbracket \Gamma \vdash \varphi_n \colon \bullet\rrbracket_\S^\B(\eta)\big)\,\Big) \\
				&\,= \heartsuit_S\big(\llbracket \Gamma \vdash \varphi_1 \colon \bullet\rrbracket_\S^\B(\eta), \dots, \llbracket \Gamma \vdash \varphi_n \colon \bullet\rrbracket_\S^\B(\eta)\big) \circ \sigma, \\
		\Big(\llbracket \Gamma \vdash \lambda X^{\tau_1}.\varphi : \tau_1 \to \tau_2 \rrbracket_\S^\B(\eta)\Big)(f) &\coloneqq \llbracket \Gamma, X : \tau_1 \vdash \varphi : \tau_2 \rrbracket_\S^\B\big(\eta[X \mapsto f]\big), \\
		&\text{where } (\eta[X \mapsto f])(X) \coloneqq f \\
		&\text{and } (\eta[X \mapsto f])(Y) \coloneqq \eta(Y) \text{ for all } Y \in \dom(\eta) \setminus\{X\}, \\
		\llbracket \Gamma \vdash \varphi \psi : \tau_2 \rrbracket_\S^\B(\eta) &\coloneqq \Big(\llbracket \Gamma \vdash \varphi \colon \tau_1 \to \tau_2 \rrbracket_\S^\B(\eta)\Big)\Big(\llbracket \Gamma \vdash \psi : \tau_1 \rrbracket_\S^\B(\eta)\Big), \\
		\llbracket \Gamma \vdash \mu X^\tau.\varphi : \tau \rrbracket_\S^\B(\eta) &\coloneqq \LFP_{S_\tau^\B}\Big(\llbracket \Gamma \vdash \lambda X^\tau.\varphi : \tau \to \tau \rrbracket_\S^\B(\eta)\Big), \\
		\llbracket \Gamma \vdash \nu X^\tau.\varphi : \tau \rrbracket_\S^\B(\eta) &\coloneqq \GFP_{S_\tau^\B}\Big(\llbracket \Gamma \vdash \lambda X^\tau.\varphi : \tau \to \tau \rrbracket_\S^\B(\eta)\Big).
	\end{align*}
	For a closed formula $\varphi$ of type $\tau$, we denote $\llbracket\varphi\rrbracket_\S^\B \coloneqq \llbracket \varnothing \vdash \varphi : \tau\rrbracket_\S^\B(\eta_\varnothing)$, where $\eta_\varnothing \colon \varnothing \to \varnothing$ is the empty function. For $s \in S$ and for a closed formula $\varphi$ of type $\bullet$, we write $\S, s \Vdash^\B \varphi$ if and only if $\llbracket \varphi \rrbracket_\S^\B(s) = \top_\B$.
\end{definition}

Let $\P \colon \Set \to \Set$ denote the covariant power set functor, and $\hat \P \colon \Set^\op \to \Set$ denote the contravariant power set functor.

\begin{example}
	Consider the complete lattice $\mathbbm 2 \coloneqq \{0, 1\}$ consisting of exactly two values, with $0 <_{\mathbbm 2} 1$. Note that $\hom_\Set(-, \mathbbm 2)$ is naturally isomorphic to the contravariant power set functor $\hat \P$, so we shall work with the latter instead.
	
	Let $\Sigma$ be a non-empty set, and let $L \colon \Set \to \Set$ be the functor defined as follows: on objects, it assigns a set $S$ to the set $LS \coloneqq (\P S)^\Sigma = \hom_\Set(\Sigma, \P S)$; on morphisms, it assigns a function $f \colon S \to S'$ to the function $Lf \colon (\P S)^\Sigma \to (\P S')^\Sigma$ defined by $(Lf)(h) \coloneqq (\P f) \circ h$. For each $a \in \Sigma$, define a unary modality $\langle a \rangle \colon \hat \P \Rightarrow \hat \P L$ by
	\[
		\langle a \rangle_S(U) \coloneqq \{\,h \in (\P S)^\Sigma \mid h(a) \cap U \neq \varnothing\,\},
	\]
	so that $\llbracket\langle a \rangle \varphi\rrbracket_\S^{\mathbbm 2} = \{\,s \in S \mid \sigma(s)(a) \cap \llbracket \varphi \rrbracket_\S^{\mathbbm 2} \neq \varnothing\,\}$ for any closed and well-typed formula $\varphi$ of type $\bullet$ and any $L$-coalgebra $\S = (S, \sigma)$.
	
	Fix an $L$-coalgebra $\S = (S, \sigma)$ and a point $s_0 \in S$. This models a labelled transition system that has: $S$ as its state space, initial state $s_0$; and a transition relation $\xrightarrow{a}\ \subseteq S \times S$ for each $a \in \Sigma$, by declaring $s \xrightarrow{a} t$ if and only if $t \in \sigma(s)(a)$. Fix $a, b \in \Sigma$ and consider the formula
	\[
		\varphi \coloneqq \Big(\nu G^{(\bullet \to \bullet) \to \bullet}.\lambda X^{\bullet\to\bullet}.\langle a \rangle\big(X(G(\lambda Y^\bullet.\langle b \rangle (XY)))\big)\Big)(\lambda Y^\bullet.\langle b \rangle Y),
	\]
	which appears in \citet*[Example~4]{KobayashiLozesBruse2017}. Then $\S,s_0 \Vdash^{\mathbbm 2} \varphi$ if and only if there exists an infinite transition path of the form
	\[
		s_0 \xrightarrow{a} s_1 \xrightarrow{b} t_{1,1} \xrightarrow{a} s_2 \xrightarrow{b} t_{2,1} \xrightarrow{b} t_{2,2} \xrightarrow{a} \cdots \underbrace{\xrightarrow{a} s_n \xrightarrow{b} t_{n,1} \xrightarrow{b} \cdots \xrightarrow{b} t_{n,n}}_{\text{one $a$ followed by $n$ instances of $b$}} \xrightarrow{a} s_{n+1} \xrightarrow{b} \cdots
	\]
	in the labelled transition system modelled by $(\S, \sigma)$ starting with $s_0$. \qed
\end{example}

\begin{example} \label{Example_Markov}
	Let $[0, 1]$ denote the unit interval in $\mathbb R$ equipped with its usual linear order. 
	
	Let $\mathrm{Prop}$ be a non-empty set of propositional variables, and let $M\colon \Set \to \Set$ be the functor defined as follows: on objects, it assigns a set $S$ to the set 
	\[
	MS \coloneqq \left\{\,\prob \in [0,1]^S \mid \sum_{s \in S} \prob(s) = 1\,\right\} \times \P(\mathrm{Prop});
	\]
	on morphisms, it assigns a function $f \colon S \to S'$ to the function $Mf \colon \{\,\prob \in [0,1]^S \mid \sum_{s \in S} \prob(s) = 1\,\} \times \P(\mathrm{Prop}) \to \{\,\prob \in [0,1]^{S'} \mid \sum_{s \in S} \prob(s) = 1\,\} \times \P(\mathrm{Prop})$ defined by
	\[
		(Mf)(\prob, U) \coloneqq \left(\left(S' \ni s' \mapsto \sum_{\substack{s \in S, \\ f(s) = s'}}\prob(s) \in [0,1]\right),\, U\right).
	\]
	For each $p \in \mathrm{Prop}$, define a nullary modality $p \colon \{*\} \Rightarrow \hom_\Set(M(-), [0,1])$ (where we have abused notation by using the same symbol $p$ for both the element of $\mathrm{Prop}$ and the modality) as follows: for each set $S$, we identify $p_S$ with the function in $\hom_\Set(MS, [0,1])$ defined by
	\[
		p_S(\prob, U) \coloneqq 
		\begin{cases}
			1, \quad &\text{if } p \in U, \\
			0, \quad &\text{otherwise}.
		\end{cases}
	\]
	Also define a unary modality $\Circle \colon \hom_\Set(-, [0,1]) \Rightarrow \hom_\Set(M(-),[0,1])$ by
	\[
		\Circle_S(f)(\prob, U) \coloneqq \sum_{s \in S} f(s)\prob(s).
	\]
	
	Fix an $M$-coalgebra $\S = (S, \sigma)$ and a point $s_0 \in S$. For $s \in S$, write $\sigma(s) = (\prob_s, U_s)$. This pointed $M$-coalgebra $(\S, s_0)$ models a Markov chain that has: state space $S$; initial state $s_0$; for each $s, t \in S$, a probability $\prob_s(t)$ of transitioning from $s$ to $t$; and for each $s \in S$, a set $U_s \subseteq \mathrm{Prop}$ of propositional variables which $s$ satisfies. Fix $p \in \mathrm{Prop}$ and consider the formula
	\[
		\varphi \coloneqq \big(\mu G^{\bullet \to \bullet}.\lambda X^\bullet.(X \lor G(\Circle X))\big)p,
	\]
	which appears in \citet*[Example~2.8]{MitaniKobayashiTsukada2021}. Then $\S,s_0 \Vdash^{[0,1]} \varphi$ if and only if
	\[
		\text{$s_0$ satisfies $p$} \quad \text{or} \quad \sup_{1 \leq n < \omega}\left(\sum_{(s_1, \dots, s_n) \in S^n} \left(p_S(\sigma(s_n))\prod_{0 \leq i \leq n-1}\prob_{s_i}(s_{i+1})\right)\right) = 1.
	\]
	That is: for all $\varepsilon > 0$ there exists $n < \omega$ such that, after $n$ transitions starting from $s_0$, the probability of landing in a state which satisfies $p$ is larger than $1 - \varepsilon$. \qed
\end{example}

\begin{remark}
    The full probabilistic HFL (PHFL) defined in \citet*{MitaniKobayashiTsukada2021} is not yet captured by our coalgebraic semantics, due to PHFL having formulas of the form $[\varphi]_J$ for an upwards-closed subset $J$ of $\B = [0, 1]$. To account for this, we will have to extend \cref{defn_syntax} to
    \[
        \varphi \Coloneqq \ \cdots \ | \ [\varphi]_J
    \]
    where $J$ ranges over the set ${{\B}{\uparrow}} \coloneqq \{\,V \subseteq \B : \text{for all } x, y \in \B, \text{ if } V \ni x \leq_\B y, \text{ then } y \in V\,\}$. \cref{defn_typing} is then extended with the typing rule
    \[
		\begin{bprooftree}
			\AxiomC{$\Gamma \vdash \varphi : \bullet$}
			\RightLabel{\scriptsize $J$}
			\UnaryInfC{$\Gamma \vdash [\varphi]_J : \bullet$}
		\end{bprooftree}
    \]
    for all $J \in {{\B}{\uparrow}}$. Finally, \cref{defn_semantics} is extended with
    \[
        \llbracket \Gamma \vdash [\varphi]_J : \bullet \rrbracket_\S^\B(\eta)(s) \coloneqq 
        \begin{cases}
            \top_\B, \quad &\text{if } \llbracket \Gamma \vdash \varphi : \bullet \rrbracket_\S^\B(\eta)(s) \in J, \\
            \bot_\B, \quad &\text{otherwise}.
        \end{cases}
    \]
\end{remark}

\begin{example}
	Let $\omega + 1$ denote the second infinite ordinal, equipped with its usual linear order.

	Let $T \colon \Set \to \Set$ be the functor defined as follows: on objects, it assigns a set $S$ to the set $TS \coloneqq (S \times S) + S + 1$, where $1 \coloneqq \{*\}$ is a singleton set; on morphisms, it assigns a function $f \colon S \to S'$ to the function $Tf \colon (S \times S) + S + 1 \to (S' \times S') + S' + 1$ defined by
	\[
		(Tf)(s) \coloneqq
		\begin{cases}
			(f(s_1), f(s_2)), \quad &\text{if } s = (s_1, s_2) \in S \times S, \\
			f(s), \quad &\text{if } s \in S, \\
			*, \quad &\text{if } s = * \in 1.
		\end{cases}
	\]
	Define unary modalities $\langle \mathtt{left} \rangle, \langle \mathtt{right} \rangle \colon \hom_\Set(-, \omega+1) \Rightarrow \hom_\Set(T(-), \omega + 1)$ by
	\begin{align*}
		\langle \mathtt{left} \rangle_S (h) (s) &\coloneqq
		\begin{cases}
			h(s_1) + 1, \quad &\text{if } s = (s_1, s_2) \in S \times S \text{ and } h(s_1) < \omega, \\
			\omega, \quad &\text{if } s = (s_1, s_2) \in S \times S \text{ and } h(s_1) = \omega, \\
			h(s) + 1, \quad &\text{if } s \in S \text{ and } h(s) < \omega, \\
			\omega, \quad &\text{if } s \in S \text{ and } h(s) = \omega, \\
			0, \quad &\text{if } s = * \in 1,
		\end{cases} \\
		\langle \mathtt{right} \rangle_S(h)(s) &\coloneqq
		\begin{cases}
			h(s_2) + 1, \quad &\text{if } s = (s_1, s_2) \in S \times S \text{ and } h(s_2) < \omega, \\
			\omega, \quad &\text{if } s = (s_1, s_2) \in S \times S \text{ and } h(s_2) = \omega, \\
			h(s) + 1, \quad &\text{if } s \in S \text{ and } h(s) < \omega, \\
			\omega, \quad &\text{if } s \in S \text{ and } h(s) = \omega, \\
			0, \quad &\text{if } s = * \in 1.
		\end{cases}
	\end{align*}
	
	Fix a $T$-coalgebra $\S = (S, \sigma)$ and a point $s_0 \in S$. This models a (finite or infinite) tree where every node is labelled by an element of $S$ and has at most two successors, and whose root is labelled by $s_0$. Consider the formula
	\[
		\varphi \coloneqq \Big(\mu G^{\bullet \to \bullet \to \bullet}.\lambda X^\bullet.\lambda Y^\bullet.\big((X \land Y) \lor G(\langle \mathtt{left} \rangle X)(\langle \mathtt{right} \rangle Y)\big)\Big)\bot\bot,
	\]
	adapted from \citet*[Section~2.1]{LangeLozesGuzman2014}. Let $n_\ell, n_r \leq \omega$ respectively denote the length (i.e. number of edges) of the left-most branch (i.e. maximal path) and the right-most branch of the tree modelled by the pointed $T$-coalgebra $(\S, s_0)$. Then $\llbracket \varphi \rrbracket_\S^{\omega+1}(s_0) = \min(n_\ell, n_r)$. \qed
\end{example}

\section{The top-value problem} \label{sec:TopValueProblem}

Throughout this section, we fix a functor $F \colon \Set \to \Set$, a complete lattice $\B$, and a set of modalities $\Mod$ for $F$ with respect to $\B$. Furthermore, we shall suppose that the functor $F$ preserves inclusions and weak wide pullbacks. In particular, if $X$ and $Y$ are sets satisfying $X \subseteq Y$, then $FX \subseteq FY$.

\begin{definition}
	For a set $X$ and for $A \in FX$, the \emph{support of $A$}, denoted $\suppt(A)$, is defined to be
	\[
		\suppt(A) \coloneqq \bigcap\{\,S \subseteq X : A \in FS\,\}.
	\]
\end{definition}

Functors which preserve inclusions and weak wide pullbacks will allow a well-behaved notion of a support, and they admit the following two canonical unary monotone predicate liftings.

\begin{definition}\label{BoxDiamond}
	Define two unary predicate liftings $\Box, \Diamond \colon \hom_\Set(-, \B) \Rightarrow \hom_\Set(F(-), \B)$ as follows: for all sets $S$, all functions $h \colon S \to \B$, and all $A \in FS$, let
	\begin{align*}
		\Box_S(h)(A) &\coloneqq \inf\{\,h(s) : s \in \suppt(A)\,\}, \text{ and} \\
		\Diamond_S(h)(A) &\coloneqq \sup\{\,h(s) : s \in \suppt(A)\,\}.
	\end{align*}
\end{definition}

The naturality of $\Box$ and $\Diamond$ in \cref{BoxDiamond} above follows from the functor $F$ preserving inclusions and weak wide pullbacks. Moving forward, we shall assume, without loss of generality, that the set $\Mod$ contains both $\Box$ and $\Diamond$.

\begin{definition}\label{FBautomaton}
	An \emph{$(F, \B)$-automaton} is a quintuple $\mathbb A = (Q, \Sigma, \Delta, q_0, \heartsuit)$ consisting of:
	\begin{enumerate}
		\item a non-empty set $Q$ of \emph{states};
		\item a non-empty set $\Sigma$ of \emph{input symbols};
		\item an $((F(-))^\Sigma \times \mathbbm 2)$-coalgebra $(Q, \Delta)$, where $\mathbbm 2 \coloneqq \{0,1\}$;
		\item an \emph{initial state} $q_0 \in Q$;
		\item and a unary modality $\heartsuit \in \Mod$.
	\end{enumerate}
	A \emph{finite $(F, \B)$-automaton} is an $(F, \B)$-automaton $\mathbb A = (Q, \Sigma, \Delta, q_0, \heartsuit)$ where both $Q$ and $\Sigma$ are finite.
	
	For an $(F, \B)$-automaton $\mathbb A = (Q, \Sigma, \Delta, q_0, \heartsuit)$, a state $q \in Q$, and a finite word $w \in \Sigma^*$, the \emph{value of the run of $w$ on $\mathbb A$ starting from $q$}, denoted $\mathbb A(q, w) \in \B$, is defined inductively on $w$ as follows:
	\begin{align*}
		\mathbb A(q, \varepsilon) &\coloneqq
		\begin{cases}
			\top_\B, \quad &\text{if $\pi_2(\Delta(q)) = 1$}, \\
			\bot_\B, \quad &\text{if $\pi_2(\Delta(q)) = 0$},
		\end{cases}  \\
		\mathbb A(q, cw) &\coloneqq \Big(\heartsuit_Q(\mathbb A(-, w))\Big)\big((\pi_1(\Delta(q)))(c)\big) \qquad \text{for $c \in \Sigma$ and $w \in \Sigma^*$},
	\end{align*}
	where $(FQ)^\Sigma \xleftarrow{\pi_1} (FQ)^\Sigma \times \mathbbm 2 \xrightarrow{\pi_2} \mathbbm 2$ are the relevant projections, and $\varepsilon$ is the empty word. We then define the \emph{value of $\mathbb A$} to be
	\[
		\mathrm{val}(\mathbb A) \coloneqq \sup_{w \in \Sigma^*} \mathbb A(q_0, w).
	\]
\end{definition}

Note that an $(F, \B)$-automaton as in \cref{FBautomaton} above is \textit{neither} a generalisation of the notion of an alternating $F$-automaton in the sense of \citet[Definition~3.1]{Venema2004} \textit{nor} that of a $\Lambda$-automaton in the sense of \citet*[Definition~4]{FontaineLealVenema2010}. Instead, \cref{FBautomaton} intends to generalise nondeterministic finite automata (NFAs) as well as the probabilistic automata introduced by \citet[Definition~4]{Rabin1963}.

Given an $(F, \B)$-automaton $\mathbb A = (Q, \Sigma, \Delta, q_0, \heartsuit)$, its \emph{top-value problem} asks whether $\mathrm{val}(\mathbb A) = \top_\B$. If $F$ is the (covariant) power set functor, $\B = \mathbbm 2$, and $\heartsuit = \Diamond$, then the associated top-value problem is simply the usual (non-)emptiness problem for NFAs. On the other hand, if $F$ is the probability distribution functor, $\B = [0, 1]$, and $\heartsuit$ is the expectation modality (i.e. the $\Circle$ modality in \citet*[Definition~2.6]{MitaniKobayashiTsukada2021}), then the associated top-value problem is the value-1 problem for probabilistic automata (see, for instance, \citet[Problem~5]{GimbertOualhadj2010}).

We generalise a result in \citet*[Theorem~3.5]{MitaniKobayashiTsukada2021} by reducing the top-value problem of a finite $(F, \B)$-automaton $\mathbb A = (Q, \Sigma, \Delta, q_0, \heartsuit)$ to the model-checking problem of an $\HFL_{\widetilde{\Mod}}^\B$ formula on a pointed $(F(-) \times \P(\Sigma+\{\cmark\}))$-coalgebra. Here, $\widetilde{\Mod}$ is a set of modalities for $F(-) \times \P(\Sigma + \{\cmark\})$ with respect to $\B$ defined as follows.
\begin{enumerate}
	\item For each $\heartsuit \in \Mod$ of arity $n \in \omega$, define $\widetilde\heartsuit \colon (\hom_\Set(-, \B))^n \Rightarrow \hom_\Set(F(-) \times \P(\Sigma +\{\cmark\}), \B)$ by
		\[
			\widetilde\heartsuit_S(h_1, \dots, h_n)(A, U) \coloneqq \heartsuit_S(h_1, \dots, h_n)(A)
		\]
		for all sets $S$, all functions $h_1, \dots, h_n \colon S \to \B$, all $A \in FS$, and all $U \subseteq \Sigma + \{\cmark\}$.
	\item For each $c \in \Sigma + \{\cmark\}$, define a nullary modality $c$ as follows: for each set $S$, we identify $c_S \colon \{*\} \to \hom_\Set(FS \times \P(\Sigma+\{\cmark\}), \B)$ with the function $c_S \colon FS \times \P(\Sigma+\{\cmark\}) \to \B$ given by
	\[
		c_S(A, U) \coloneqq
		\begin{cases}
			\top_\B, \quad &\text{if } c \in U, \\
			\bot_\B, \quad &\text{otherwise},
		\end{cases}
	\]
	for all $A \in FS$ and all $U \subseteq \Sigma + \{\cmark\}$.
	\item Let $\widetilde{\Mod} \coloneqq \{\,\widetilde\heartsuit : \heartsuit \in \Mod\,\} \cup (\Sigma + \{\cmark\})$.
\end{enumerate}
In essence, $(F(-) \times \P(\Sigma +\{\cmark\}))$-coalgebras are $F$-coalgebras which allow propositional constants in $(\Sigma + \{\cmark\})$. Here, $\cmark$ is a new symbol which we will later use to label a state in $\mathbb A$ which is accepting. The set $\widetilde \Mod$ is then the set of modalities consisting of all the modalities in $\Mod$, appropriately modified for $(F(-) \times \P(\Sigma +\{\cmark\}))$-coalgebras, together with one nullary modality for each propositional constant in $(\Sigma + \{\cmark\})$.

\begin{theorem} \label{TopValueProblem_to_HFL}
	Let $\mathbb A = (Q, \Sigma, \Delta, q_0, \heartsuit)$ be a finite $(F, \B)$-automaton. Furthermore, suppose that for each $q \in Q$ there exists $A_q \in F(Q \times \Sigma)$ such that $\suppt(A_q) = \{q\} \times \Sigma$. Define an $(F(-) \times \P(\Sigma + \{\cmark\}))$-coalgebra $\mathbb S = (S, \sigma)$ as follows:
	\begin{align*}
		S &\coloneqq Q + (Q \times \Sigma), \\
		\sigma(s) &\coloneqq 
		\begin{cases}
			(A_q,\, \{\cmark\}), \quad &\text{if } s = q \in Q \text{ and } \pi_2(\Delta(q)) = 1, \\
			(A_q,\, \varnothing), \quad &\text{if } s = q \in Q \text{ and } \pi_2(\Delta(q)) = 0, \\
			\big((\pi_1(\Delta(q)))(c),\, \{c\}\big), \quad &\text{if } s = (q, c) \in Q \times \Sigma,
		\end{cases}
	\end{align*}
	where $(FQ)^\Sigma \xleftarrow{\pi_1} (FQ)^\Sigma \times \mathbbm 2 \xrightarrow{\pi_2} \mathbbm 2$ are the relevant projections. Then
	\[
		\left\llbracket \left(\mu G^{\bullet \to \bullet}.\lambda X^\bullet.\left(X \lor \bigvee_{c \in \Sigma} G(\widetilde\Diamond(c \land \widetilde\heartsuit X))\right)\right)\cmark \right\rrbracket_\S^\B(q_0) = \mathrm{val}(\mathbb A).
	\]
	In particular, $\mathrm{val}(\mathbb A) = \top_\B$ if and only if $\S, q_0 \Vdash^\B \Big(\mu G^{\bullet \to \bullet}.\lambda X^\bullet.\big(X \lor \bigvee_{c \in \Sigma} G(\widetilde\Diamond(c \land \widetilde\heartsuit X))\big)\Big)\cmark$.
\end{theorem}
\begin{proof}
	For $c \in \Sigma$, define the formula
	\[
		f_c \coloneqq \lambda X^\bullet.\widetilde\Diamond(c \land \widetilde\heartsuit X),
	\]
    of type $\bullet \to \bullet$, so that
	\begin{align*}
		&\left\llbracket\mu G^{\bullet \to \bullet}.\lambda X^\bullet.\left(X \lor \bigvee_{c \in \Sigma} G(\widetilde\Diamond(c \land \widetilde\heartsuit X))\right)\right\rrbracket_\S^\B \\
		&= \left\llbracket \mu G^{\bullet \to \bullet}.\lambda X^\bullet.\left(X \lor \bigvee_{c \in \Sigma} G(f_c X)\right) \right\rrbracket_\S^\B \\
		&= \sup_{n < \omega} \left\llbracket \lambda X^\bullet.\left(X \lor \bigvee_{c_1 \in \Sigma} \left(f_{c_1} X \lor \bigvee_{c_2 \in \Sigma} \left(f_{c_1}f_{c_2}X \lor \left(\dots \lor \bigvee_{c_n \in \Sigma} f_{c_1} \dots f_{c_n} X \right) \dots \right)\right)\right) \right\rrbracket_\S^\B \\
		&= \sup_{n < \omega} \sup_{w \in \Sigma^{\leq n}}\left\llbracket \lambda X^\bullet.(f_{w_1}(f_{w_2}(\dots(f_{w_{|w|}}X)\dots)))\right\rrbracket_\S^\B,
	\end{align*}
	where we have written $w = w_1 \dots w_{|w|}$ for $w \in \Sigma^{\leq n}$, and the suprema above are taken in $S_{\bullet \to \bullet}^\B$. This gives us
	\[
		\left\llbracket \left(\mu G^{\bullet \to \bullet}.\lambda X^\bullet.\left(X \lor \bigvee_{c \in \Sigma} G(\widetilde\Diamond(c \land \widetilde\heartsuit X))\right)\right)\cmark \right\rrbracket_\S^\B(q_0) = \sup_{n < \omega} \sup_{w \in \Sigma^{\leq n}} \llbracket f_{w_1}(f_{w_2}(\dots (f_{w_{|w|}} \cmark) \dots )) \rrbracket_\S^\B(q_0).
	\]
	
	We will now establish, by an induction on $n < \omega$, that the following equality holds for all states $q \in Q$ and all words $w = w_1 \dots w_n \in \Sigma^{n}$:
	\[
		\llbracket f_{w_1}(\dots(f_{w_n}\cmark)\dots) \rrbracket_\S^\B(q) = \mathbb A(q, w).
	\]
	The equality is obvious when $w$ is the empty word, by definition of the value $\mathbb A(q, \varepsilon)$. Suppose inductively that $\llbracket f_{w_1'}(\dots(f_{w_n'}\cmark)\dots) \rrbracket_\S^\B(q) = \mathbb A(q, w')$ for all states $q \in Q$ and a word $w' = w_1'\dots w_n' \in \Sigma^n$. Then for all $c \in \Sigma$, we obtain the following string of equalities:
	\begin{align*}
		&\llbracket f_c(f_{w_1'}(\dots(f_{w_n'}\cmark)\dots)) \rrbracket_\S^\B(q) \\
		&= \llbracket \widetilde\Diamond(c \land \widetilde\heartsuit(f_{w_1'}(\dots(f_{w_n'}\cmark)\dots))) \rrbracket_\S^\B(q) \\
		&= \Big(\widetilde\Diamond_S\big(\llbracket c \land \widetilde\heartsuit(f_{w_1'}(\dots(f_{w_n'}\cmark)\dots))\rrbracket_\S^\B\big)\Big)(\sigma(q)) \\
		&= \Big(\widetilde\Diamond_S\big(\llbracket c \land \widetilde\heartsuit(f_{w_1'}(\dots(f_{w_n'}\cmark)\dots))\rrbracket_\S^\B\big)\Big)(A_q, U_q), \qquad \text{where $U_q \coloneqq \{\cmark\}$ if $\pi_2(\Delta(q)) = 1$, else $U_q \coloneqq \varnothing$,} \\
		&= \sup\{\,\llbracket c \land \widetilde\heartsuit(f_{w_1'}(\dots(f_{w_n'}\cmark)\dots))\rrbracket_\S^\B(q', c') : (q', c') \in \suppt(A_q)\,\} \\
		&= \sup\{\,\llbracket c \land \widetilde\heartsuit(f_{w_1'}(\dots(f_{w_n'}\cmark)\dots))\rrbracket_\S^\B(q, c') : c' \in \Sigma\,\} \\
		&= \llbracket \widetilde\heartsuit (f_{w_1'}(\dots(f_{w_n'}\cmark)\dots)) \rrbracket_\S^\B(q, c) \\
		&= \widetilde\heartsuit_S \big(\llbracket f_{w_1'}(\dots(f_{w_n'}\cmark)\dots) \rrbracket_\S^\B\big)((\pi_1(\Delta(q)))(c),\, \{c\}) \\
		&= \heartsuit_S\big(\llbracket f_{w_1'}(\dots(f_{w_n'}\cmark)\dots) \rrbracket_\S^\B\big)((\pi_1(\Delta(q)))(c)) \\
		&= \heartsuit_Q\big(\llbracket f_{w_1'}(\dots(f_{w_n'}\cmark)\dots) \rrbracket_\S^\B \restriction_Q\big)((\pi_1(\Delta(q)))(c)) \\
		&= \heartsuit_Q(\mathbb A(-, w'))((\pi_1(\Delta(q)))(c)) \\
		&= \mathbb A(q, cw'),
	\end{align*}
	where the third-to-last equality follows from the naturality of $\heartsuit$ and the assumption that $F$ preserves inclusions, and the second-to-last equality follows from the inductive hypothesis.
	
	Therefore
	\begin{align*}
		\left\llbracket \left(\mu G^{\bullet \to \bullet}.\lambda X^\bullet.\left(X \lor \bigvee_{c \in \Sigma} G(\widetilde\Diamond(c \land \widetilde\heartsuit X))\right)\right)\cmark \right\rrbracket_\S^\B(q_0) 
		&= \sup_{n < \omega} \sup_{w \in \Sigma^{\leq n}} \llbracket f_{w_1}(f_{w_2}(\dots (f_{w_{|w|}} \cmark) \dots )) \rrbracket_\S^\B(q_0) \\
		&= \sup_{n < \omega} \sup_{w \in \Sigma^{\leq n}} \mathbb A(q_0, w) \\
		&= \sup_{w \in \Sigma^*} \mathbb A(q_0, w) \\
		&= \mathrm{val}(\mathbb A),
	\end{align*}
    as desired.
\end{proof}

In addition to preserving inclusions and weak wide pullbacks, the hypotheses of \cref{TopValueProblem_to_HFL} require the functor $F$ to satisfy the following property:
\begin{quote}
	``for any non-empty finite sets $Q$ and $\Sigma$, and for any $q \in Q$, there exists $A_q \in F(Q \times \Sigma)$ such that $\suppt(A_q) = \{q\} \times \Sigma$''.
\end{quote}
Since there exist functors $F \colon \Set \to \Set$ which preserve inclusions and weak wide pullbacks but do not satisfy the property above, the following \cref{CanonicalObjectWithFiniteSupport} provides a sufficient condition for the property above to be satisfied. It involves the non-empty finite power set functor $\P_{\mathrm{fin},\neq\varnothing} \colon \Set \to \Set$, which is the subfunctor of the (covariant) power set functor $\P$ defined on objects by
	\[
		\P_{\mathrm{fin},\neq\varnothing}(X) \coloneqq \{\,S \subseteq X : 0 < |S| < \aleph_0\,\}
	\]
	for all sets $X$.

\begin{proposition} \label{CanonicalObjectWithFiniteSupport}
	Suppose there exists a monic natural transformation $\kappa \colon \P_{\mathrm{fin},\neq\varnothing} \Rightarrow F$. Then for any set $X$ and for any non-empty finite subset $S \subseteq X$, we have $\suppt(\kappa_X(S)) = S$.
\end{proposition}
\begin{proof}
    Fix any (non-empty) set $X$ as well as a non-empty finite subset $\varnothing \neq S \subseteq X$. We wish to show that $\suppt(\kappa_X(S)) = S$. For sets $S'$ and $X'$ with $S' \subseteq X'$, let $\iota_{S'}^{X'} \colon S' \hookrightarrow X'$ denote the inclusion function of $S'$ into $X'$. For a function $f' \colon X' \to Y'$, let $f'[S']$ denote the image of $S'$ under $f'$.

    The commuting square on the left below
\[\begin{tikzcd}
	{\P_{\mathrm{fin},\neq\varnothing} S} &&&& {\P_{\mathrm{fin},\neq\varnothing} X} && S && S \\
	\\
	FS &&&& FX && {\kappa_S(S)} && {\kappa_X(S)}
	\arrow["{\iota_S^X[-]}", hook, from=1-1, to=1-5]
	\arrow["{{{{\kappa_S}}}}"', tail, from=1-1, to=3-1]
	\arrow["{{{{\kappa_X}}}}", tail, from=1-5, to=3-5]
	\arrow["{\iota_S^X[-]}", maps to, from=1-7, to=1-9]
	\arrow["{{\kappa_S}}"', maps to, from=1-7, to=3-7]
	\arrow["{{{{\kappa_X}}}}", maps to, from=1-9, to=3-9]
	\arrow["{{{{F\iota_S^X \, = \, \iota_{FS}^{FX}}}}}"', hook, from=3-1, to=3-5]
	\arrow["{{{{\iota_{FS}^{FX}}}}}"', maps to, from=3-7, to=3-9]
\end{tikzcd}\]
    in $\Set$ gives rise to the chase of elements on the right. Since $\iota_{FS}^{FX}$ is the inclusion function of $FS$ into $FX$, it follows that $\kappa_X(S) = \kappa_S(S) \in FS$. So $\suppt(\kappa_X(S)) \subseteq S$.

    Now suppose we are given any subset $T \subseteq X$ such that $\kappa_X(S) \in FT$. We want to show that $S \subseteq T$. We split into two cases: when $T = \varnothing$ and when $T \neq \varnothing$.

    If $T = \varnothing$, then $\kappa_X(S) \in F\varnothing$. Since $X$ is non-empty by assumption, we can define two different functions $f_0, f_1 \colon X \to \mathbbm 2$ by $f_i(x) \coloneqq i$ for all $x \in X$ and for each $i \in \mathbbm 2 \coloneqq \{0, 1\}$. The initiality of $\varnothing$ in $\Set$ gives us $f_0 \iota_\varnothing^X = f_1 \iota_\varnothing^X$, and thus $(F f_0)(F\iota_\varnothing^X) = (Ff_1)(F\iota_\varnothing^X)$, yielding $(Ff_0)(\kappa_X(S)) =  (Ff_1)(\kappa_X(S))$. The commuting naturality squares
\[\begin{tikzcd}
	{\P_{\mathrm{fin},\neq\varnothing} X} && {\P_{\mathrm{fin},\neq\varnothing}\mathbbm 2} \\
	\\
	FX && {F \mathbbm 2}
	\arrow["{f_i[-]}", from=1-1, to=1-3]
	\arrow["{\kappa_X}"', tail, from=1-1, to=3-1]
	\arrow["{\kappa_{\mathbbm 2}}", tail, from=1-3, to=3-3]
	\arrow["{Ff_i}"', from=3-1, to=3-3]
\end{tikzcd}\]
    in $\Set$, for each $i \in \mathbbm 2$, then imply that 
    \[
        \kappa_{\mathbbm{2}} (\{0\}) = \kappa_{\mathbbm 2}(f_0[\kappa_X(S)]) = \kappa_{\mathbbm 2}(f_1[\kappa_X(S)]) = \kappa_{\mathbbm 2}(\{1\}).
    \]
    The above equality contradicts the injectivity of $\kappa_{\mathbbm 2}$. So the case $T = \varnothing$ is impossible.

    On the other hand, if $T \neq \varnothing$, then we can fix some $t \in T$ and define a function $f \colon X \to X$ by
    \[
        f(x) \coloneqq 
        \begin{cases}
            x, \quad &\text{if } x \in T, \\
            t, \quad &\text{otherwise},
        \end{cases}
    \]
    for all $x \in X$, so that $f \iota_T^X = \iota_T^X$. Then $(Ff)(F\iota_T^X) = F(\iota_T^X)$, and so $(Ff)(\kappa_X(S)) = \kappa_X(S)$. The naturality of $\kappa$ then yields $\kappa_X(f[S]) = \kappa_X(S)$. As $\kappa_X$ is injective, we obtain $f[S] = S$. By definition of $f$, we have $f[S] \subseteq T$, and therefore $S \subseteq T$.
\end{proof}

It may be interesting to note that, if the objects $A_q$ in the hypotheses of \cref{TopValueProblem_to_HFL} are obtained from the setting of \cref{CanonicalObjectWithFiniteSupport} by letting $A_q \coloneqq \kappa_{Q \times \Sigma}(\{q\} \times \Sigma)$, then the construction of $(S, \sigma)$ from $(Q, \Delta)$ as in the hypotheses of \cref{TopValueProblem_to_HFL} actually defines, on objects, a colimit-preserving functor from the category of $((F(-))^\Sigma \times \mathbbm 2)$-coalgebras to the category of $(F(-) \times \P(\Sigma + \{\cmark\}))$-coalgebras. In fact, this functor has a right adjoint. More concretely, we have the following \cref{automata_and_transition_systems} about this functor.

\begin{proposition}\label{automata_and_transition_systems}
	Fix a non-empty finite set $\Sigma$ and a monic natural transformation $\kappa \colon \P_{\mathrm{fin},\neq\varnothing} \Rightarrow F$. Define the following four functors.
	\begin{enumerate}
		\item Define a functor $L \colon \Coalg((F(-))^\Sigma \times \mathbbm 2) \to \Coalg(F(-) \times (\Sigma + \mathbbm 2))$ as follows.
			\begin{enumerate}
				\item For an $((F(-))^\Sigma \times \mathbbm 2)$-coalgebra $(X, \xi)$, let $L(X, \xi)$ be the $(F(-) \times (\Sigma + \mathbbm 2))$-coalgebra with state space $X + (X \times \Sigma)$ and structure map $\gamma_{(X,\xi)} \colon X + (X \times \Sigma) \to F(X + (X \times \Sigma)) \times (\Sigma + \mathbbm 2)$ given by
					\begin{alignat*}{2}
						\gamma_{(X,\xi)}(x) &\coloneqq (\kappa_{X \times \Sigma}(\{x\} \times \Sigma),\, \pi_2(\xi(x))), \quad &&\text{for all } x \in X, \\
						\gamma_{(X,\xi)}(x,c) &\coloneqq ((\pi_1(\xi(x)))(c),\,c),\quad &&\text{for all } (x,c) \in X \times \Sigma,
					\end{alignat*}
					where $(FX)^\Sigma \xleftarrow{\pi_1} (FX)^\Sigma \times \mathbbm 2 \xrightarrow{\pi_2} \mathbbm 2$ are the relevant projections.
				\item For an $((F(-))^\Sigma \times \mathbbm 2)$-coalgebra morphism $h$, define $Lh \coloneqq h + (h \times \id_\Sigma)$.
			\end{enumerate}
		\item Define a functor $R \colon \Coalg(F(-) \times (\Sigma + \mathbbm 2)) \to \Coalg((F(-))^\Sigma \times \mathbbm 2)$ as follows.
			\begin{enumerate}
				\item For an $(F(-) \times (\Sigma + \mathbbm 2))$-coalgebra $(Y, \gamma)$, define a function $\Phi \colon \P Y \to \P Y$ by
				\begin{align*}
					\Phi(U) \coloneqq 
						\Big\{\,y \in Y \ : \ 
							&\pi_2(\gamma(y)) \in \mathbbm 2, \\
							&\pi_1(\gamma(y)) = \kappa_Y(S_y) \text{ for some (unique) } S_y \in \P_{\mathrm{fin},\neq\varnothing}(Y) \text{ with } |S_y| = |\Sigma|, \\
							&\text{and for all } c \in \Sigma \text{ there exists (a unique) } y_c \in S_y \text{ such that } \\
							&\pi_1(\gamma(y_c)) \in FU \text{ and } \pi_2(\gamma(y_c)) = c\,\Big\}
				\end{align*}
				for all subsets $U \subseteq Y$, where $FY \xleftarrow{\pi_1} FY \times (\Sigma + \mathbbm 2) \xrightarrow{\pi_2} \Sigma + \mathbbm 2$ are the relevant projections. Then let $R(Y, \gamma)$ be the $((F(-))^\Sigma \times \mathbbm 2)$-coalgebra whose state space $\overline Y$ is the greatest fixed-point of $\Phi$ (where $\P Y$ is equipped with its usual order) and whose structure map $\xi_{(Y,\gamma)} \colon \overline Y \to (F\overline Y)^\Sigma \times \mathbbm 2$ is given by
				\[
					\xi_{(Y, \gamma)}(y) \coloneqq ((\Sigma \ni c \mapsto \pi_1(\gamma(y_c)) \in F\overline Y),\,\pi_2(\gamma(y)))
				\]
				for all $y \in \overline Y$.
				\item For an $(F(-) \times (\Sigma + \mathbbm 2))$-coalgebra morphism $h$, let $Rh$ be the function $h \restriction_{\overline{\dom(h)}}$.
			\end{enumerate}
		\item Define a functor $H \colon \Coalg(F(-) \times (\Sigma + \mathbbm 2)) \to \Coalg(F(-) \times \P(\Sigma + \{\cmark\}))$ as follows.
			\begin{enumerate}
				\item For an $(F(-) \times (\Sigma + \mathbbm 2))$-coalgebra $(Y, \gamma)$, let $H(Y, \gamma)$ be the $(F(-) \times \P(\Sigma + \{\cmark\}))$-coalgebra with state space $Y$ and structure map $\zeta_{(Y, \gamma)} \colon Y \to FY \times \P(\Sigma + \{\cmark\})$ defined by
					\[
						\zeta_{(Y,\gamma)}(y) \coloneqq 
						\begin{cases}
							(\pi_1(\gamma(y)),\, \{\pi_2(\gamma(y))\}), \quad &\text{if } \pi_2(\gamma(y)) \in \Sigma, \\
							(\pi_1(\gamma(y)),\, \{\cmark\}), \quad &\text{if } \pi_2(\gamma(y)) = 1 \in \mathbbm 2, \\
							(\pi_1(\gamma(y)),\, \varnothing), \quad &\text{if } \pi_2(\gamma(y)) = 0 \in \mathbbm 2,
						\end{cases}
					\]
					for all $y \in Y$, where $FY \xleftarrow{\pi_1} FY \times (\Sigma + \mathbbm 2) \xrightarrow{\pi_2} \Sigma + \mathbbm 2$ are the relevant projections.
				\item For an $(F(-) \times (\Sigma + \mathbbm 2))$-coalgebra morphism $h$, define $Hh \coloneqq h$.
			\end{enumerate}
		\item Finally, define a functor $K \colon \Coalg(F(-) \times \P(\Sigma + \{\cmark\})) \to \Coalg(F(-) \times (\Sigma + \mathbbm 2))$ as follows.
			\begin{enumerate}
				\item For an $(F(-) \times \P(\Sigma + \{\cmark\}))$-coalgebra $(Z, \zeta)$, define a function $\Psi \colon \P Z \to \P Z$ by
					\[
						\Psi(V) \coloneqq \{\,z \in Z : \pi_1(\zeta(z)) \in FV \text{ and } |\pi_2(\zeta(z))| \leq 1\,\}
					\]
					for all subsets $V \subseteq Z$, where $FZ \xleftarrow{\pi_1} FZ \times \P(\Sigma + \{\cmark\}) \xrightarrow{\pi_2} \P(\Sigma + \{\cmark\})$ are the relevant projections. Then let $K(Z, \zeta)$ be the $(F(-) \times (\Sigma + \mathbbm 2))$-coalgebra whose state space $\widetilde Z$ is the greatest fixed-point of $\Psi$ and whose structure map $\gamma_{(Z,\zeta)} \colon \widetilde Z \to F\widetilde Z \times (\Sigma + \mathbbm 2)$ is given by
					\[
						\gamma_{(Z,\zeta)}(z) \coloneqq 
						\begin{cases}
							(\pi_1(\zeta(z)), \, c), \quad &\text{if } \pi_2(\zeta(z)) = \{c\} \text{ for some } c \in \Sigma, \\
							(\pi_1(\zeta(z)), \, 1), \quad &\text{if } \pi_2(\zeta(z)) = \{\cmark\}, \\
							(\pi_1(\zeta(z)), \, 0), \quad &\text{if } \pi_2(\zeta(z)) = \varnothing,
						\end{cases}
					\]
					for all $z \in \widetilde Z$.
				\item For an $(F(-) \times \P(\Sigma + \{\cmark\}))$-coalgebra morphism $h$, let $Kh$ be the function $h \restriction_{\widetilde{\dom(h)}}$.
			\end{enumerate}
	\end{enumerate}
	Then we have the following adjunctions:
\[\begin{tikzcd}
	{\Coalg((F(-))^\Sigma \times \mathbbm 2)} & {\Coalg(F(-) \times (\Sigma + \mathbbm 2))} & {\Coalg(F(-) \times \P(\Sigma + \{\cmark\})).}
	\arrow[""{name=0, anchor=center, inner sep=0}, "L", shift left=2, from=1-1, to=1-2]
	\arrow[""{name=1, anchor=center, inner sep=0}, "R", shift left=2, from=1-2, to=1-1]
	\arrow[""{name=2, anchor=center, inner sep=0}, "H", shift left=2, from=1-2, to=1-3]
	\arrow[""{name=3, anchor=center, inner sep=0}, "K", shift left=2, from=1-3, to=1-2]
	\arrow["\dashv"{anchor=center, rotate=-90}, draw=none, from=0, to=1]
	\arrow["\dashv"{anchor=center, rotate=-90}, draw=none, from=2, to=3]
\end{tikzcd}\]
\end{proposition}
\begin{proof}[Proof sketch]
	Proving this is tedious but elementary.
	
	The fact that the mapping $L$ sends $((F(-))^\Sigma \times \mathbbm 2)$-coalgebra morphisms to $(F(-) \times (\Sigma + \mathbbm 2))$-coalgebra morphisms is due to the naturality of $\kappa$ as well as the ambient assumption that $F$ preserves inclusions. Its functoriality is clear.
	
	For the mapping $R$, we can show that an $(F(-) \times (\Sigma + \mathbbm 2))$-coalgebra morphism $h \colon (Y, \gamma) \to (Y', \gamma')$ restricts to a function $h \restriction_{\overline Y} \colon \overline Y \to \overline{Y'}$ by transfinite induction on the construction of $\overline{Y'}$ from Kleene's fixed-point theorem. The monotonicity of the function $\Phi$ uses the assumption that $F$ preserves inclusions. Furthermore, we use the naturality of $\kappa$ as well as the assumption that $F$ preserves inclusions to perform the inductive step of the argument. That $h \restriction_{\overline Y} \colon \overline Y \to \overline{Y'}$ is actually an $((F(-))^\Sigma \times \mathbbm 2)$-coalgebra morphism from $(\overline Y, \xi_{(Y,\gamma)})$ to $(\overline{Y'}, \xi_{(Y', \gamma')})$ is easy to see. The functoriality of $R$ is also clear.
	
	It is obvious that $H$ is a functor. We can prove that $K$ is a functor in a similar manner to how we proved that $R$ is a functor.
	
	The adjunction bijection for $L \dashv R$ is as follows. For an $(F(-) \times (\Sigma+\mathbbm 2))$-coalgebra morphism $f \colon L(X,\xi) \to (Y, \gamma)$, its transpose is $\overline f \colon (X,\xi) \to R(Y, \gamma)$ given by $\overline{f}(x) \coloneqq f(x)$ for all $x \in X$. In the other direction, for an $((F(-))^\Sigma \times \mathbbm 2)$-coalgebra morphism $g \colon (X,\xi) \to R(Y, \gamma)$, its transpose is $\overline g \colon L(X, \xi) \to (Y, \gamma)$ given by $\overline{g}(x) \coloneqq g(x)$ for all $x \in X$, and $\overline{g}(x, c) \coloneqq (g(x))_c$ for all $(x,c) \in X \times \Sigma$.
	
	The adjunction bijection for $H \dashv K$ is easier. For an $(F(-) \times \P(\Sigma + \{\cmark\}))$-coalgebra morphism $f \colon H(Y, \gamma) \to (Z, \zeta)$, its transpose $\overline f \colon (Y, \gamma) \to K(Z, \zeta)$ is defined by $\overline f(y) \coloneqq f(y)$ for all $y \in Y$. Finally, for an $(F(-) \times (\Sigma + \mathbbm 2))$-coalgebra morphism $g \colon (Y, \gamma) \to K(Z, \zeta)$, its transpose $\overline g \colon H(Y, \gamma) \to (Z, \zeta)$ is defined by $\overline g(y) \coloneqq g(y)$ for all $y \in Y$.
\end{proof}

\section{Future work}

We close with several possible directions that can be taken with this formulation of coalgebraic HFL.

\begin{question}
    Does the fixed-point alternation hierarchy for $\HFL_\Mod^\B$ collapse?
\end{question}

This question has been attacked with some success by \citet[Corollary~6.2.12]{BrusePhDthesis2018} when the functor in question is the (covariant) power set functor, $\B = \mathbbm 2$, and when our attention is restricted to order-1 formulas (also see \cite{Bruse2016}, but be aware of the footnote on page 147 of \cite{BrusePhDthesis2018}). In this scenario, the alternation of fixed-point operators has been shown to produce a strict hierarchy.

\begin{question}
    What conditions ensure the decidability of the model-checking problem for $\HFL_\Mod^\B$?
\end{question}

If the state space of the coalgebra in question is finite, the functor in question is the power set functor, and the complete lattice in question is the two-valued complete lattice $\mathbbm 2$, then the model-checking problem is clearly decidable. In particular, the model-checking problem for the original formulation of HFL \citep{ViswanathanViswanathan2004} is decidable. However, in \citet*[Corollary~3.6]{MitaniKobayashiTsukada2021}, it is shown that when both the functor and complete lattice in question are as in \cref{Example_Markov}, then the model-checking problem becomes undecidable. This is because \citet[Theorem~4]{GimbertOualhadj2010} showed that the value-1 problem for probabilistic automata is undecidable.

\begin{question}
    Return to the setting of \cref{sec:TopValueProblem}. Let $\mathbb A = (Q, \Sigma, \Delta, q_0, \heartsuit)$ be a finite $(F, \B)$-automaton. Define the \emph{language of $\mathbb A$}, denoted $\mathcal L(\mathbb A)$, to be
    \[
        \mathcal L(\mathbb A) \coloneqq \{\,w \in \Sigma^* : \mathbb A(q_0, w) = \top_\B\,\}.
    \]
    Can the problem of the (non-)emptiness of $\mathcal L(\mathbb A)$ be reduced to a model-checking problem of some coalgebraic HFL formula?
\end{question}

A condition that makes the above possible is when the lattice $\B$ satisfies the following property:
\[
    \text{for any $S \subseteq \B$, if $\sup S = \top_\B$ then $\top_\B \in S$.}
\]
This sufficient condition is evidently rather unsatisfactory.

\renewcommand \refname{References}
\addcontentsline{toc}{section}{References}
\begin{flushleft}
	\bibliography{biblio.bib}
\end{flushleft}
\bibliographystyle{plainnat}
\nocite{*}

\end{document}